\documentclass[10pt]{article}
\usepackage[latin1]{inputenc}
\usepackage{amsmath,amssymb,amsthm}
\usepackage{graphicx}
\usepackage{epsfig}
\newtheorem{theo}{Theorem}[section]

\newtheorem{defi}[theo]{Definition}

\newcounter{listagem}
\newcommand{\blista}{\begin{list}{\roman{listagem})}{\usecounter{listagem}}}
\newcommand{\elista}{\end{list}}
\newcommand{\BR}{{\mathbb R}}

\newcommand{\BC}{{\mathbb C}}

\newcommand{\tr}{\mathop{\rm tr}}

\newcommand{\beq}{\begin{equation}}
\newcommand{\eeq}{\end{equation}}
\newcommand{\beqn}{\begin{eqnarray}}
\newcommand{\eeqn}{\end{eqnarray}}
\newcommand{\ov}{\overline}

\newcommand{\Cl}{{C \kern -0.1em \ell}}

\newcommand{\f}{\mathfrak{f}}
\newcommand{\fm}{\mathfrak{f}^{\dagger}}

\newcommand{\ds}{\displaystyle}

\begin{document}

\title{Factorization of the Non-Stationary Schr\"odinger Operator}

\author{Paula Cerejeiras ~~~~ Nelson Vieira\\
{\small Department of Mathematics,}\\{\small University of Aveiro,}\\{\small 3810-193 Aveiro, Portugal.} \\
{\small E-mails: pceres@mat.ua.pt, ~nvieira@mat.ua.pt}} \maketitle

\begin{abstract}
We consider a factorization of the non-stationary Schr\"odin\-ger
ope\-rator based on the parabolic Dirac operator introduced by
Cerejeiras/ K\"ahler/ Sommen. Based on the fundamental solution for
the pa\-rabolic Dirac operators, we shall construct appropriated
Teodorescu and Cauchy-Bitsadze operators. Afterwards we will
describe how to solve the nonlinear Schr\"odinger equation using
Banach fixed point theorem. \\

{\bf Keywords:}Nonlinear PDE's, Parabolic Dirac operators, Iterative
Methods

{\bf MSC 2000:} Primary: 30G35; Secundary: 35A08, 15A66.

\end{abstract}

\maketitle
\section{Introduction}

Time evolution problems are of extreme importance in mathematical
physics. However, there is still a need for special techniques to
deal with these problems, specially when non-linearities are
involved.

For stationary problems, the theory developed by K. G\"urlebeck and
W. Spr\"o\ss ig \cite{GS}, based on an orthogonal decomposition of
the underlying function space in terms of the subspace of
null-solutions of the corresponding Dirac operator, has been
successfully applied to a wide range of equations, for instance
Lam\'e, Navier-Stokes, Maxwell or Schr\"odinger equations \cite{CK},
\cite{GS}, \cite{KK}, \cite{D} or \cite{ShapKrav}. Unfortunately, there is no easy way to extend
this theory directly to non-stationary problems.

In \cite{CKS} the authors proposed an alternative approach in terms
of a Witt basis. This approach allowed a successful application of
the already existent techniques of elliptic function theory (see
\cite{GS}, \cite{CK}) to non-stationary problems in time-varying
domains. Namely, a suitable orthogonal decomposition for the
underlying function space was obtained in terms of the kernel of the
parabolic Dirac operator and its range after application to a
 Sobolev space with zero boundary-values.

In this paper we wish to apply this approach to study the existence
and uniqueness of solutions of the non-stationary nonlinear
Schr\"odinger equation.

Initially, in section two, we will present some basic notions about
complexified Clifford algebras and Witt basis. In section three we
will present a factorization for the operators $(\pm i \partial_t
-\Delta )$ using an extension of the parabolic Dirac operator
introduced in \cite{CKS}. For the particular case of the
non-stationary Schr\"odinger operator we will present the
corresponding Teodorescu and Cauchy-Bitsadze operators  in analogy
to  \cite{GS}. Moreover, we will obtain some direct results about
the decomposition of $L_p$-spaces  and the resolution of the linear
Schr\"odinger problem.

In the last section we will present an algorithm to solve
numerically the non-linear Schr\"odinger and we prove its
convergence in $L_2$-sense using Banach's fixed point theorem.


\section{Preliminaries}

We consider the $m$-dimensional vector space $\BR^m$ endowed with an
orthonormal basis $\{ e_1, \cdots, e_m\}.$

We define the universal Clifford algebra $\Cl_{0,m}$ as the
$2^{m}$-dimensional associative algebra which preserves the
multiplication rules $e_i e_j + e_j e_i = -2\delta_{i,j}.$ A basis
for $\Cl_{0,m}$ is given by $ e_{0} = 1$ and $ e_A = e_{h_1} \cdots
e_{h_k},$ where  $A = \{ h_1, \ldots, h_k \} \subset M = \{ 1,
\ldots, m \}$, for $1 \leq h_1 < \cdots < h_k \leq m$. Each element
$x \in \Cl_{0,m}$ will be represented by $x=\sum_{A} x_A e_A,$ $x_A
\in \BR,$ and each non-zero vector $x = \sum_{j=1}^m x_j e_j\in
\BR^{m}$ has a multiplicative inverse given by $\frac{-x}{|x|^2}$.
We denote by $\ov{x}^{\Cl_{0,m}}$ the (Clifford) conjugate of the
element $x \in \Cl_{0,m},$ where
    \begin{eqnarray*}
    \ov{1}^{\Cl_{0,m}}=1,  ~ \ov{e_j}^{\Cl_{0,m}} = -e_j,  ~ \ov{ab}^{\Cl_{0,m}}=\ov{b}^{\Cl_{0,m}}\ov{a}^{\Cl_{0,m}}.
    \end{eqnarray*}

We introduce the complexified Clifford algebra $\Cl_m$ as the
tensorial product
    \begin{eqnarray*}
    \mathbb{C} \otimes \Cl_{0,m} = \left \{ w=\sum_{A} z_A e_A , ~ z_A \in
    \mathbb{C}, A \subset M \right \}
    \end{eqnarray*}
where the imaginary unit interacts with the basis elements via $i
e_j = e_j i, j= 1, \ldots, m.$ The conjugation in
$\Cl_m = \mathbb{C} \otimes \Cl_{0,m}$ will be defined as
$\ov{w} = \sum_{A} \ov{z_A}^{\mathbb{C}} \ov{e_A}^{\Cl_{0,m}}.$ Let
us remark that for $a,b \in \Cl_m$ we have $|ab| \leq 2^m |a||b|$.

We introduce the Dirac operator $D=\sum_{j=1}^{m} e_j
\partial_{x_i}$. It factorizes the $m$-di\-men\-sional Laplacian, that
is, $D^2=-\Delta$. A $\Cl_m$-valued function defined on an open
domain $\underline \Omega,$ $u:\underline \Omega \subset \BR^m
\mapsto\Cl_m,$ is said to be {\it left-monogenic} if it satisfies
$Du=0$ on $\underline \Omega$ (resp. \textit{right-monogenic} if it
satisfies $uD=0$ on $\underline \Omega$).

A function $u:\underline \Omega \mapsto \Cl_{m}$ has a
representation $u=\sum_A u_A e_A$ with $\BC $-valued components
$u_A$. Properties such as continuity will be understood
component-wisely. In the following we will use the short notation
$L_p(\underline \Omega)$, $C^k(\underline \Omega)$, etc., instead of
$L_p(\underline \Omega,\Cl_m)$, $C^k(\underline \Omega,\Cl_m)$. For
more details on Clifford analysis, see \cite{DSS}, \cite{McIM}, \cite{BDSVA} or \cite{GM}.

Taking into account \cite{CKS} we will imbed $\mathbb{R}^m$ into
$\mathbb{R}^{m+2}$. For that purpose we add two new basis elements
$\f$ and $\fm$ satisfying
    $$\begin{array}{c}
         {\f}^2 = {\fm}^2 = 0, ~~  \f \fm + \fm \f = 1, ~~
         \f e_j +e_j \f =  \fm e_j +e_j \fm = 0, j=1,\cdots,m.
    \end{array}$$
This construction will allows us to use a suitable factorization
of the time evolution operators where only partial derivatives are
used.


\section{Factorization of time-evolution operators}

In this section we will study the forward/backward Schr\"odinger
equations,
    \begin{equation}
   (\pm i \partial_t -\Delta ) u(x,t)=0, ~~(x,t) \in \Omega,
   \label{Eq:1}
    \end{equation} where $\Omega \subset \mathbb{R}^m \times \mathbb{R}^+$, $m \geq
3,$ stands for an open domain in $\BR^m \times \BR^+$. We remark at
this point that $\Omega$ is a time-variating domain and, therefore,
not necessarily a cylindric domain.

Taking account the ideas presented in \cite{B} and \cite{CKS} we
introduce the following definition
\begin{defi}
For a function $u \in W_p^1(\Omega), ~1<p<+\infty,$ we define the
forward (resp. backward) parabolic Dirac operator
    \begin{eqnarray}
    D_{x,\pm i t}u & = & (D + \f \partial_t \pm i \fm)u, \label{I}
    \end{eqnarray}
where $D$ stands for the (spatial) Dirac operator.
\end{defi}

It is obvious that $D_{x,\pm it}:W_p^1(\Omega)\rightarrow
L_p(\Omega)$.

These operators factorize the correspondent time-evolution operator
 (\ref{Eq:1}), that is
    \begin{eqnarray}
    (D_{x, \pm it})^2 u = (\pm i \partial_t-\Delta) u. \label{II}
    \end{eqnarray}

Moreover, we consider the generic Stokes' Theorem
\begin{theo}\label{Th:22}
For each $u, v \in W_p^1(\Omega),$ $1<p<\infty,$ it holds
    \begin{eqnarray*}
    \int_{\Omega} v d \sigma_{x,t} u & = & \int_{\partial \Omega}  [(v D_{x, -it}) u + v ( D_{x, +it} u )]
    dx dt
    \end{eqnarray*}
where the surface element is $d \sigma _{x,t} = (D_x + \f
\partial_t) \rfloor dxdt,$ the contraction of the homogeneous
operator associated to $D_{x, -it}$ with the volume element.
\end{theo}

We now construct the fundamental solution for the time-evolution
operator $-\Delta - i \partial_t$. For that purpose, we consider the
fundamental solution of the heat operator
    \begin{eqnarray}
    e(x,t) & = & \frac{H(t)}{(4 \pi t)^\frac{m}{2}} \exp \left(
    {-\frac{|x|^2}{4t}} \right), \label{Form:1}
    \end{eqnarray}
where $H(t)$ denotes the Heaviside-function. Let us remark that the
previous fundamental solution verifies
    \begin{eqnarray*}
    (-\Delta + \partial_t) e(x,t) & = & \delta(x) \delta(t).
    \end{eqnarray*}

We apply to (\ref{Form:1}) the rotation $t \rightarrow it$. There we
obtain
$$
    (-\Delta - i \partial_t) e(x,it) = -\Delta  e(x,it) +
    \partial_{it} e(x,it) =   \delta(x) \delta(it) = -i \delta(x) \delta(t),$$
i.e., the fundamental solution for the Schr\"odinger operator
$-\Delta - i \partial_t$ is
    \begin{eqnarray}
    e_-(x,t) & = & i e(x,it) \nonumber \\
    & = & i ~ \frac{H(t)}{(4 \pi i t)^\frac{m}{2}} \exp \left(i ~{\frac{|x|^2}{4t}}
    \right). \label{Form:2}
    \end{eqnarray}

Then we have
\begin{defi}
    Given the fundamental solution $e_-=e_-(x,t)$ we have as fundamental solution $E_-=E_-(x,t)$ for
    the parabolic Dirac operator $D_{x,-it}$ the function
    \begin{eqnarray}
    E_-(x,t) & = &  e_-(x,t) D_{x,-it} \nonumber \\
    & = & \frac{H(t)}{(4\pi i t)^{\frac{m}{2}}} \exp{\left(\frac{i|x|^2}{4t} \right)}
    \left(\frac{-x}{2t} + \f \left(  \frac{|x|^2}{4t^2} - \frac{im}{2t} \right) + \fm
    \right) \label{SFD}
    \end{eqnarray}
\end{defi}

If we replace the function $v$ by the fundamental solution $E_-$ in
the generic Stoke's formula presented before, we have, for a
function $u \in W_p^1(\Omega)$ and a point $(x_0,t_0) \notin
\partial \Omega$, the Borel-Pompeiu formula,
    \begin{gather}
    \int_{\partial \Omega} E_-(x-x_0,t-t_0) d\sigma_{x,t} \,  u(x,t) \nonumber \\
    = u(x_0,t_0) + \int_{\Omega} E_-(x-x_0,t-t_0) (D_{x,+it}u) dx
    dt.     \label{III}
    \end{gather}

Moreover, if $u \in ker(D_{x,+it})$ we obtain the Cauchy's integral
formula
    $$\int_{\partial \Omega} E_-(x-x_0,t-t_0) d\sigma_{x,t} \, \, u(x,t) =
     u(x_0,t_0).$$

Based on expression (\ref{III}) we define the Teodo\-rescu and
Cauchy-Bitsadze operators.

\begin{defi} For a function $u \in L_p(\Omega)$ we have

(a) the Teodorescu operator
    \begin{eqnarray}
    T_- u(x_0,t_0) & = & \int_{\Omega} E_-(x-x_0,t-t_0)u(x,t)dx dt \label{genT}
    \end{eqnarray}

(b) the Cauchy-Bitsadze operator
    \begin{eqnarray}
        F_- u(x_0,t_0) & = &
    \int_{\partial \Omega} E_-(x-x_0,t-t_0)d\sigma_{x,t}      u(x,t), \label{genF}
    \end{eqnarray}
for $(x_0,t_0) \notin \partial \Omega.$
\end{defi}

Using the previous operators, (\ref{III}) can be rewritten as
    $$F_- u = u +  T_-D_{x,+it}u,$$
whenever $v \in W_p^1(\Omega),~1<p<\infty.$

Moreover, the Teodurescu operator is the right inverse of the
parabolic Dirac operator $D_{x,-it)}$, that is,
\begin{eqnarray*}
D_{x,-it} T u & = & \int_{\Omega} D_{x,-it}
E_-(x-x_0,t-t_0)u(x,t)dx dt \\
& = & \int_{\Omega} \delta(x-x_0,t-t_0)u(x,t)dx dt \\
& = & u(x_0,t_0),
\end{eqnarray*} for all $ (x_0,t_0) \in \Omega.$

In view of the previous definitions and relations, we obtain the following  results, in an
analogous way as in \cite{CKS}.

\begin{theo} If $v \in W_p^{\frac{1}{2}}(\partial \Omega )$ then the trace of the
operator $F_-$ is
    \begin{eqnarray}
    \tr (F_-v) & = & \frac{1}{2}v - \frac{1}{2}S_-v, \label{S_-}
    \end{eqnarray}
where
    \begin{eqnarray*}
    S_-v(x_0,t_0) & = &     \int_{\partial {\Omega} } E_-(x-x_0,t-t_0) d\sigma_{x,t}
    v(x,t)
    \end{eqnarray*}
is a generalization of the Hilbert transform.
\end{theo}

Also, the operator $S_-$ satisfies $S_-^2=I$ and, therefore, the
operators $$\mathbf{P} =\frac{1}{2} I + \frac{1}{2}S_-, ~~\mathbf{Q}
=\frac{1}{2} I - \frac{1}{2}S_-$$ are projections into the Hardy
spaces.

Taking account the ideas presented in \cite{CKS} an immediate
application is given by the decomposition of the $L_p-$space.
\begin{theo}
The space $L_p(\Omega)$, for $1<p \leq 2$, allows the following
decomposition
    $$L_p(\Omega)  =  L_p(\Omega) \cap \textrm{ker} \left( D_{x,- it}\right) \oplus
    D_{x, it} \left( \stackrel{\circ}{W}_p^1
    (\Omega)\right),$$
and we can define the following projectors
    \begin{eqnarray*}
    P_-: & L_p(\Omega) & \rightarrow L_p(\Omega) \cap \textrm{ker}\left(
    D_{x,- it}\right) \\
    Q_- : & L_p(\Omega) & \rightarrow D_{x,- it} \left(
    \stackrel{\circ}{W}_p^1(\Omega) \right).
    \end{eqnarray*}
\end{theo}
\begin{proof}
Let us denote by $(-\Delta -i\partial_t)_0^{-1}$ the solution
operator of the problem
    $$\left \{
    \begin{array}{rcl}
    \ds (- \Delta - i \partial_t)u & = & f ~~  in \, \, \, \Omega \\
    & & \\
    u & = & 0 ~~  on \, \, \, \partial \Omega
    \end{array}
    \right. $$

As a first step we take a look at the intersection of the two
subspaces $D_{x,- it} \left( \stackrel{\circ}{W}_p^1
(\Omega)\right)$ and $L_p(\Omega) \cap \textrm{ker} \left( D_{x,-
it}\right)$.

Consider $u \in ~L_p(\Omega) \cap \textrm{ker} \left( D_{x,-
it}\right) \cap D_{x,- it} \left(
\stackrel{\circ}{W}_p^1(\Omega)\right)$. It is immediate that
$D_{x,- it}u=0$ and also, because $u \in
D_{x,-it}\left(\stackrel{\circ}{W}_p^1(\Omega)\right)$, there exist
a function $v \in \stackrel{\circ}{W}_p^1(\Omega)$ with $D_{x,- it}
v = u$ and $(-\Delta - i\partial_t)v = 0$.

Since $(-\Delta -i\partial_t)_0^{-1}f$ is unique (see \cite{V}) we
get $v=0$ and, consequently, $u=0$, i. e., the intersection of this
subspaces contains only the zero function. Therefore, our sum is a
direct sum.

Now let us $u \in L_p(\Omega)$. Then we have
$$u_2 = D_{x,-it}
(-\Delta -i\partial_t)^{-1}_0D_{x,-it}u ~ \in ~ D_{x,- it} \left(
\stackrel{\circ}{W}_p^1(\Omega)\right).$$

Let us now apply $D_{x,-it}$ to the function $u_1=u-u_2$. This
results in
    \begin{eqnarray*}
    D_{x,-it} u_1 & = & D_{x,-it}u - D_{x,-it}u_2 \\
    & = & D_{x,-it}u - D_{x,-it}D_{x,-it}
    (-\Delta -i\partial_t)^{-1}_0D_{x,-it}u \\
    & = & D_{x,-it}u - (-\Delta - i\partial_t)
    (-\Delta -i\partial_t)^{-1}_0D_{x,-it}u \\
    & = & D_{x,-it}u - D_{x,-it}u \\
    & = & 0,
    \end{eqnarray*}
i.e., $D_{x,-it} u_1 \in \textrm{ker} \left( D_{x,- it}\right)$.
Because $u \in L_p(\Omega)$ was arbitrary chosen our decomposition
is a decomposition of the space $L_p(\Omega)$.
\end{proof}

In a similar way we can obtain a decomposition of the $L_p(\Omega)$
space in terms of the parabolic Dirac operator $ D_{x,+it}.$
Moreover, let us remark that the above decompositions are orthogonal
in the case of $p=2.$

Using the previous definitions we can also present an immediate
application in the resolution of the linear Schr\"odinger problem
with homogeneous boundary data.

\begin{theo} Let $f \in L_p(\Omega), ~1<p \leq 2.$ The solution of the problem
    $$\left \{
    \begin{array}{rcl}
 (- \Delta -i \partial_t)u & = & f ~~  in ~~~ \Omega \\
    & & \\
    u & = & 0 ~~  on ~~~ \partial \Omega
    \end{array}
    \right. $$
is given by $u=T_-Q_-T_- f.$
\end{theo}
\begin{proof}
The proof of this theorem is based on the properties of the operator
$T_-$ and of the projector $Q_-$. Because $T_-$ is the right inverse
of $D_{x,-it}, $ we get
    $$D_{x,-it}^2 u = D_{x,-it} (Q_-T_- f) = D_{x,-it} (T_- f) =
    f.$$
\end{proof}


\section{The Non-Linear Schr\"odinger Problem}

In this section we will construct an iterative method for the
non-linear Schr\"odinger equation and we study is convergence. As
usual, we consider the $L_2-$norm
    $$||f||^2 = \int_\Omega [f \overline f]_0  dxdt,$$
where $[\cdot ]_0$ denotes the scalar part.

Moreover, we also need the mixed Sobolev spaces
$W_p^{\alpha,\beta}(\Omega)$. For this we introduce the convention
    \begin{eqnarray*}
    \Omega^t = \left\{ x:~(x,t) \in \Omega \right\} & \subset &
    \BR^m \\
    \Omega^x = \left\{ t:~(x,t) \in \Omega \right\} & \subset &
    \BR^+.
    \end{eqnarray*}

Then, we say that
    \begin{eqnarray*}
    u \in W_p^{\alpha,\beta}(\Omega) & \mbox{\rm ~ iff ~} & \left\{ \begin{array}{cc}
                                               u(\cdot,t) \in W_p^\alpha(\Omega^t), & ~ \forall ~ t  \\
                                               & \\
                                               u(x,\cdot) \in W_p^\beta(\Omega^x), & ~ \forall ~  x
                                             \end{array}
                                             \right.
    \end{eqnarray*}

Under this conditions we will study the (generalized) non-linear Schr\"odinger
problem:
    \begin{eqnarray}
    -\Delta_x u -i \partial_t u + |u|^2 u & = & f   ~~ \mbox{\rm ~ in ~}     \Omega \label{Eq:2} \\
    u & = & 0    ~~ \mbox{\rm ~ on ~}    \partial \Omega , \nonumber
    \end{eqnarray}
where $|u|^2=\sum_A |u_A|^2.$ We can rewrite (\ref{Eq:2}) as
    \begin{eqnarray}
    D_{x,-it}^2 u + M(u) & = & 0, \label{2*}
    \end{eqnarray}
where $M(u)=|u|^2u-f$. It is easy to see that
    \begin{eqnarray}
    u & = & - T_-Q_-T_-(M(u)) \label{2*'}
    \end{eqnarray} is a solution of (\ref{2*}) by means of direct
    application of $D_{x,-it}^2$ to both sides of the equation.

We remark that for $u \in W^{2,1}_2(\Omega)$, we get
$$
    ||D_{x,-it}u|| =||Q_-T_-M(u)||  = ||T_-M(u)||. $$

We now prove that (\ref{2*'}) can be solved by the convergent
iterative method
    \begin{eqnarray}
    u_n & = & - T_-Q_-T_-(M(u_{n-1})). \label{3*}
    \end{eqnarray}

For that purpose we need to establish some norm estimations.
Initially, we have that
    \begin{eqnarray}
    ||u_n-u_{n-1}|| & = & ||T_-Q_-T_-[M(u_{n-1})-M(u_{n-2})]|| \nonumber \\
    & \leq &  C_1 ||M(u_{n-1})-M(u_{n-2})||, \label{4*}
    \end{eqnarray}
where $C_1=||T_-Q_-T_-||=||T_-||^2$.

We now estimate the factor $||M(u_{n-1})-M(u_{n-2})||.$ We get
    \begin{gather}
    ||M(u_{n-1})-M(u_{n-2})||= |||u_{n-1}|^2u_{n-1}-|u_{n-2}|^2u_{n-2}|| \nonumber \\
\leq |||u_{n-1}|^2(u_{n-1}-u_{n-2})|| + |||u_{n-1}-u_{n-2}|^2u_{n-2}|| \nonumber \\
 \leq 2^{m+1} ||u_{n-1}-u_{n-2}|| \left( ||u_{n-1}||^2 + ||u_{n-2}|| ||u_{n-1}-u_{n-2}|| \right), \nonumber
    \end{gather}

We assume $\mathcal{K}_n:=2^{m+1} \left( ||u_{n-1}||^2 + ||u_{n-2}||
||u_{n-1}-u_{n-2}|| \right)$ so that $$    ||u_{n} -u_{n-1}||
 \leq C_1 \mathcal{K}_n ||u_{n-1}-u_{n-2}||. $$

Moreover, we have additionally that
    \begin{eqnarray}
    ||u_n|| & = & ||T_-Q_-T_-M(u_{n-1})|| \nonumber \\
    & \leq & 2^{m+1}C_1||u_{n-1}||^3 + C_1||f|| \label{5*}
    \end{eqnarray} holds.

In order to prove that indeed we have a contraction we need to study
the auxiliary inequality
$$2^{m+1}C_1 ||u_{n-1}||^3 +C_1||f||  \leq ||u_{n-1}||, $$ that
    is,
    \begin{eqnarray}
||u_{n-1}||^3 -\frac{||u_{n-1}||}{2^{m+1}C_1} +\frac{||f||}{2^{m+1}}
& \leq & 0. \label{6*}
    \end{eqnarray}

The analysis of (\ref{6*}) will be made considering two cases

\textbf{Case I: } When $||u_{n-1}||\geq 1$, we can establish the
following
    inequality in relation to (\ref{6*})
        \begin{eqnarray*}
        ||u_{n-1}||^2 -\frac{||u_{n-1}||}{3\cdot 2^{m+1}} + \frac{||f||}{2^{m+1}} & \leq &
        ||u_{n-1}||^3 -\frac{||u_{n-1}||}{2^{m+1}C_1} +\frac{||f||}{2^{m+1}}.
        \end{eqnarray*}
    Then, from (\ref{6*}), we have
        \begin{gather}
       ||u_{n-1}||^2 -\frac{||u_{n-1}||}{3 \cdot 2^{m+1}} + \frac{||f||}{2^{m+1}}  \leq 0 \nonumber \\
         ||u_{n-1}||^2 -2\frac{||u_{n-1}||}{6 \cdot 2^{m+1}} + \frac{1}{36 \cdot 2^{2m+2}} + \frac{||f||}{2^{m+1}} - \frac{1}{36 \cdot 2^{2m+2}} \leq 0 \nonumber \\
        \left( ||u_{n-1}||  - \frac{1}{6 \cdot 2^{m+1}}   \right)^2  \leq \, \frac{1}{36 \cdot 2^{2(m+1)}}
        -\frac{||f||}{2^{m+1}} \nonumber \\
        = \frac{1}{2^{m+1}} \left(\frac{1}{36 \cdot 2^{m+1}}-||f|| \right).        \label{7*}
        \end{gather}

    If $||f|| \leq \frac{1}{36 \cdot 2^{m+1}}$ then
        \begin{eqnarray*}
        \left| ||u_{n-1}||  - \frac{1}{6 \cdot 2^{m+1}} \right| & \leq & W,
        \end{eqnarray*}
    where $W=\sqrt{\frac{1}{36 \cdot 2^{2(m+1)}}-\frac{||f||}{2^{m+1}}}$.

    In consequence, if
        \begin{eqnarray*}
        \frac{1}{6 \cdot 2^{m+1}}-W & \leq \,  ||u_{n-1}|| & \leq \,
        \frac{1}{6 \cdot 2^{m+1}}+W
        \end{eqnarray*} then we have from (\ref{5*}) the desired inequality $$||u_{n}|| \leq ||u_{n-1}||.$$

    Furthermore, we have now to study the remaining case. Assuming now that $||u_{n-1}|| \leq
    \frac{1}{6 \cdot 2^{m+1}}-W$, we have
 $$
        ||u_{n}||  \leq   2^{m+1}C_1 \left( \frac{1}{6 \cdot 2^{m+1}} - W
        \right)^3+ C_1 ||f||   \leq  \frac{1}{6 \cdot 2^{m+1}} - W$$
    and $        ||u_{n-1}|| \leq \frac{1}{6 \cdot 2^{m+1}} - W,$ $||u_{n-2}|| \leq \frac{1}{6 \cdot 2^{m+1}}
        -W $ so that it holds $$
        ||u_{n-1}-u_{n-2}|| \leq 2\left( \frac{1}{6 \cdot 2^{m+1}} - W
        \right).$$

    With the previous relations we can estimate the value of $\mathcal{K}_n$
        \begin{eqnarray}
        \mathcal{K}_n & = & 2^{m+1}\left( ||u_{n-1}||^2 + ||u_{n-2}|| ||u_{n-1}-u_{n-2}||
        \right) \nonumber \\
        & \leq & 2^{m+1} \left[ \left( \frac{1}{6 \cdot 2^{m+1}} - W
        \right)^2 + 2 \left( \frac{1}{6 \cdot 2^{m+1}} - W
        \right)^2 \right]\nonumber  \\
        & \leq & 3 \cdot 2^{m+1}\left( \frac{1}{6 \cdot 2^{m+1}} - W
        \right) \nonumber \\
        & = & \frac{1}{2} - 3 \cdot 2^{m+1} W  < \frac{1}{2}, \label{8*}
        \end{eqnarray}
    which implies that
        \begin{eqnarray*}
        ||u_{n-2}|| & \leq & R:=\frac{1}{3 \cdot 2^{m+1}}.
        \end{eqnarray*}

    Finally, we have that
        \begin{eqnarray*}
        ||u_n-u_{n-1}|| & \leq & \mathcal{K}_n ||u_{n-1}-u_{n-2}||,
        \end{eqnarray*}
    with $\mathcal{K}_n < \frac{1}{2}.$\\

\textbf{Case II: }
    When $||u_{n-1}|| <  1$, we can establish the following inequality
        \begin{eqnarray*}
        ||u_{n-1}||^4 -\frac{||u_{n-1}||^2}{3 \cdot 2^{m+1}} + \frac{||f||}{2^{m+1}} & \leq &
        ||u_{n-1}||^3 -\frac{||u_{n-1}||}{2^{m+1}C_1} +\frac{||f||}{2^{m+1}}.
        \end{eqnarray*}
    Then, from (\ref{6*}), we have
        \begin{eqnarray}
        & ||u_{n-1}||^4 -\frac{||u_{n-1}||^2}{3 \cdot 2^{m+1}} + \frac{||f||}{2^{m+1}} & \leq \,
        0 \nonumber \\
        & & \nonumber \\
        \Leftrightarrow & \left( ||u_{n-1}||^2  - \frac{1}{6 \cdot 2^{m+1}}
        \right)^2 & \leq \, \frac{1}{36 \cdot 2^{2m+2}}-\frac{||f||}{2^{m+1}}.
        \label{10*}
        \end{eqnarray}

    Again, if $||f|| \leq \frac{1}{36 \cdot 2^{m+1}}$ then
        \begin{eqnarray*}
        \left| ||u_{n-1}||^2  - \frac{1}{6 \cdot 2^{m+1}} \right| & \leq & W,
        \end{eqnarray*}
    where $W=\sqrt{\frac{1}{36 \cdot 2^{2m+2}}-\frac{||f||}{2^{m+1}}}$.

    As a consequence,
        $$\begin{array}{cccc}
        &  \frac{1}{6 \cdot 2^{m+1}}-W & \leq \,  ||u_{n-1}||^2 &         \leq \, \frac{1}{6 \cdot 2^{m+1}}+W \\
        & & & \\
        \Leftrightarrow & \sqrt{\frac{1}{6 \cdot 2^{m+1}}-W} & \leq \,  ||u_{n-1}|| &         \leq \, \sqrt{\frac{1}{6 \cdot 2^{m+1}}+W}
        \end{array}$$ leads to $||u_n|| \leq ||u_{n-1}||.$

    Again, considering now the case of $||u_{n-1}|| \leq
    \sqrt{\frac{1}{6 \cdot 2^{m+1}}-W}$, we obtain
        \begin{eqnarray*}
        ||u_{n}|| & \leq \, 2^{m+1}C_1 \left(\sqrt{\frac{1}{6 \cdot 2^{m+1}} - W}\right)^3 + C_1 ||f|| & \leq \, \sqrt{\frac{1}{6 \cdot 2^{m+1}} - W }
        \end{eqnarray*}
    and
        $$\begin{array}{c}
        ||u_{n-1}|| \leq \sqrt{\frac{1}{6 \cdot 2^{m+1}} - W }
         \qquad  ||u_{n-2}|| \leq \sqrt{\frac{1}{6 \cdot 2^{m+1}} - W }
        -W \\ \\
        ||u_{n-1}-u_{n-2}|| \leq 2\sqrt{\frac{1}{6 \cdot 2^{m+1}} - W}.
        \end{array}$$

    With the previous relations we can estimate the value of $\mathcal{K}_n$
        \begin{eqnarray}
        \mathcal{K}_n & = & 2^{m+1}\left( ||u_{n-1}||^2 + ||u_{n-2}|| ||u_{n-1}-u_{n-2}||
        \right) \nonumber \\
        & \leq & 2^{m+1} \left[ \left( \frac{1}{6 \cdot 2^{m+1}} - W
        \right) + 2 \left( \frac{1}{6 \cdot 2^{m+1}} - W
        \right) \right]\nonumber  \\
        & = & 3 \cdot 2^{m+1}\left( \frac{1}{6 \cdot 2^{m+1}} - W
        \right) \nonumber \\
        & = & \frac{1}{2} - 3 \cdot 2^{m+1} W  < \frac{1}{2}, \label{11*}
        \end{eqnarray}
    which implies that
        \begin{eqnarray*}
        ||u_{n-2}|| & \leq & R:=\frac{1}{3 \cdot 2^{m+1}}.
        \end{eqnarray*}

    Finally, we have that
        \begin{eqnarray*}
        ||u_n-u_{n-1}|| & \leq & \mathcal{K}_n ||u_{n-1}-u_{n-2}||,
        \end{eqnarray*}
    with $\mathcal{K}_n < \frac{1}{2}.$\\

The application of Banach's fixed point, to the previous
conclusions, results in the following theorem
\begin{theo}
The problem (\ref{Eq:2}) has a unique solution $u \in \,
W_2^{2,1}(\Omega)$ if $f \in L_2(\Omega)$ satisfies the condition
    \begin{eqnarray*}
    ||f|| & \leq & \frac{1}{36 \cdot 2^{m+1}}.
    \end{eqnarray*}
Moreover, our iteration method (\ref{3*}) converges for each
starting point\linebreak $u_0 \in \stackrel{\circ}{W}_2^{1,1}(\Omega)$ such
that
    \begin{eqnarray*}
    ||u_0|| & \leq & \frac{1}{6 \cdot 2^{m+1}} + W,
    \end{eqnarray*}
with $W = \sqrt{\frac{1}{36 \cdot 2^{2(m+1)}}-\frac{||f||}{2^{m+1}}}$.
\end{theo}


{\bf Acknowledgement} {\it The second author wishes to express his
gratitude to {\it Funda\c c\~ao para a Ci\^encia e a Tecnologia} for
the support of his work via the grant {\tt SFRH/BSAB/495/2005}.}



\end{document}